\documentclass[english]{article}
\usepackage[T1]{fontenc}
\usepackage[latin9]{inputenc}
\usepackage{amsmath}
\usepackage{amsthm}
\usepackage{amssymb}
\usepackage{wasysym}

\makeatletter

\newcommand{\lyxmathsym}[1]{\ifmmode\begingroup\def\b@ld{bold}
  \text{\ifx\math@version\b@ld\bfseries\fi#1}\endgroup\else#1\fi}

\newcommand{\lyxaddress}[1]{
	\par {\raggedright #1
	\vspace{1.4em}
	\noindent\par}
}
\newenvironment{lyxlist}[1]
	{\begin{list}{}
		{\settowidth{\labelwidth}{#1}
		 \setlength{\leftmargin}{\labelwidth}
		 \addtolength{\leftmargin}{\labelsep}
		 }}
	{\end{list}}
\theoremstyle{plain}
\newtheorem{thm}{\protect\theoremname}
\theoremstyle{definition}
\newtheorem{defn}[thm]{\protect\definitionname}

\@ifundefined{date}{}{\date{}}
\usepackage[all]{xy}

\usepackage{placeins}
\usepackage{lineno}
\usepackage{flushend}
\usepackage{graphicx}
\usepackage{changepage}
\usepackage{cite}
\usepackage[titletoc,title]{appendix}
\usepackage{titlesec}

\titleformat{\subsection}[runin]
  {\normalfont\bfseries}{\thesubsection}{1pt}{.}
\titleformat{\subsubsection}[runin]
  {\normalfont\bfseries}{\thesubsubsection}{1em}{}

\makeatother

\usepackage{babel}
\providecommand{\definitionname}{Definition}
\providecommand{\theoremname}{Theorem}

\begin{document}
\title{The Contextuality-by-Default View of the Sheaf-Theoretic Approach
to Contextuality}
\author{Ehtibar N. Dzhafarov}
\maketitle

\lyxaddress{\begin{center}
\textsuperscript{}Purdue University, ehtibar@purdue.edu\\
 
\par\end{center}}
\begin{abstract}
The Sheaf-Theoretic Contextuality (STC) theory developed by Abramsky
and colleagues is a very general account of whether multiply overlapping
subsets of a set, each of which is endowed with certain ``local''
structure, can be viewed as inheriting this structure from a global
structure imposed on the entire set. A fundamental requirement of
STC is that any intersection of subsets inherit one and the same structure
from all intersecting subsets. I show that when STC is applied to
systems of random variables, it can be recast in the language of the
Contextuality-by-Default (CbD) theory, and this allows one to extend
STC to arbitrary systems, in which the requirement in question (called
``consistent connectedness'' in CbD) is not necessarily satisfied.
When applied to possibilistic systems, such as systems of logical
statements with unknown truth values, the problem arises of distinguishing
lack of consistent connectedness from contextuality. I show that it
can be resolved by considering systems with multiple possible deterministic
realizations as quasi-probabilistic systems with epistemic (or Bayesian)
probabilities assigned to the realizations. Although STC and CbD have
distinct native languages and distinct aims and means, the conceptual
modifications presented in this paper seem to make them essentially
coextensive.\medskip{}

\end{abstract}

\section{Introduction}

\subsection{}

Contextuality-by-Default (CbD) and Sheaf-Theoretic Contextuality (STC)
are two general approaches to establishing and measuring (non)contextuality
of systems of measurements. The word ``measurement'' is understood
very broadly, including various relations between inputs and outputs
in a physical entity, between databases and records, between logical
statements and their truth values, etc. The two theories employ distinct
mathematical languages, and were designed with different aims in mind.
With respect to their applicability areas, when dealing with probabilisitic
scenarios, STC is confined to special systems, called ``strongly
consistently connected'' in CbD. By contrast, CbD was designed with
the primary purpose to apply to arbitrary probabilistic systems. On
the other hand, CbD is confined to probabilistic systems, with all
deterministic systems being trivially noncontextual. By contrast,
STC offers contextuality analysis of systems that are inherently deterministic,
finding there interesting cases of contextual systems. These comparative
characterizations are incomplete, and I will elaborate them as we
proceed. 

\subsection{}

The two theories are represented by numerous publications, and their
representative exposition can be found in Refs. \cite{AbramBarbMans2017,AbramskyBrand2011,Abramsky2014,Abramskyetal2017,AbtramskyBrandSavochking2014}
for STC and Refs. \cite{DzhCerKuj2017,KUjDzh2019Measures,KujDzhProof2016,DzhKujCer_Polytope,DzhKuj2017Fortsch}
for CbD. Familiarity with the two theories would be helpful in reading
this paper, but all relevant results and concepts, including those
mentioned in this introduction, will be explained below. The reader
is to be warned, however, that I will not endeavor here what might
be a worthy project for future work: to systematically present the
two theories using their native languages and establish correspondences
and differences between them. Rather, as the title of the paper suggests,
STC will only be presented here from the point of view of CbD, in
the CbD language. 

\subsection{}

I will show that by doing this one can easily extend STC, including
the important notion of (non)contextual fraction, to apply to inconsistently
connected systems. More precisely, the CbD language allows one to
redefine inconsistently connected systems into consistently connected
ones, so that STC can apply to them.

\subsection{}

I will argue that there is a conceptual problem in applying original
STC, with its commitment to strong consistent connectedness, to inherently
deterministic systems, such as systems of statements with definitive
truth values. A way of dealing with this issue I propose is to distinguish
completely specified deterministic systems (that are always noncontextual)
and systems with multiple possible deterministic realizations. In
the latter case, by assigning Bayesian priors to these realizations
one renders such systems quasi-probabilistic, disentangling thereby
contextuality and inconsistent connectedness. This construction allows
one to extend CbD to the contextuality analysis of deterministic systems
in the spirit of STC. 

\section{CbD: Conceptual set-up}

\subsection{}

In CbD, the object of contextuality analysis is a \emph{system of
random variables} representing what generically can be called measurements.
Depending on application, the random variables may describe outputs
of physical measurements, responses to inputs, answers to questions,
etc. The random variables in a system are assumed to be \emph{dichotomous}
(say, +1/-1), for reasons discussed in Sections \ref{subsec:firstdreason}
and \ref{subsec:secondreason}. This means that any measurement is
presented as simultaneous answers to a set of Yes/No questions, such
as ``Is the measured value less than 5?''. 

\subsection{}

The question a random variable answers is referred to as the \emph{content}
of the random variable. A set of random variables forms a system if
they are labeled both by their contents and by their \emph{contexts}.
A context of a variable includes all conditions recorded \emph{together}
with this variable, where ``together'' can mean any empirical procedure
by which the observed values of the random variables and conditions
within the context are paired. For instance, if a random variable
$R$ is recorded together with two other random variables, this fact
is part of the context of $R$. If the order of recording these random
variables is itself systematically recorded, this is part of the context
of $R$ too.

\subsection{}

The terminology is: if content $q$ is measured in context $c$, which
is written 
\[
q\prec c,
\]
the outcome of the measurement is the (dichotomous) random variable
$R_{q}^{c}$. 

\subsection{}

If the sets of the contents and contexts are finite (\emph{as we are
going to assume throughout this paper}), the system can be represented
by a matrix like this:

\begin{equation}
\begin{array}{|c|c|c|c||c|}
\hline \begin{array}{c}
\\
\\
\end{array}R_{1}^{1} & R_{2}^{1} &  &  & c^{1}\\
\hline \begin{array}{c}
\\
\\
\end{array} & R_{2}^{2} & R_{3}^{2} & R_{4}^{2} & c^{2}\\
\hline \begin{array}{c}
\\
\\
\end{array}R_{1}^{3} &  & R_{3}^{3} &  & c^{3}\\
\hline \begin{array}{c}
\\
\\
\end{array}R_{1}^{4} &  &  & R_{4}^{4} & c^{4}\\
\hline \begin{array}{c}
\\
\\
\end{array}R_{1}^{5} & R_{2}^{5} & R_{3}^{5} &  & c^{5}\\
\hline\hline \begin{array}{c}
\\
\\
\end{array}q_{1} & q_{2} & q_{3} & q_{4} & \mathcal{R}
\\\hline \end{array}\:.\label{eq:exmaple1}
\end{equation}
This system has four contents variously measured in five contexts,
and $R_{i}^{j}$ is the abbreviation for $R_{q=q_{i}}^{c=c^{j}}$.
I will use this system as an example throughout the paper. 

\subsection{}

The following are the two basic properties of any system. 
\begin{lyxlist}{00.00.0000}
\item [{CbD1:}] All random variables sharing a context are jointly distributed
(i.e., they can be presented as measurable functions on one and the
same probability space).
\item [{CbD2:}] Any two random variables in different contexts are stochastically
unrelated, i.e., they are defined on distinct probability spaces.
\end{lyxlist}
\subsection{}

If any two random variables in the system that have the same content
are identically distributed, writing this as 
\[
R_{q}^{c}\sim R_{q}^{c'},
\]
the system is called (\emph{simply}) \emph{consistently connected}.
CbD does not assume this property: generally, a system can be\emph{
inconsistently connected }(and this term is often used in the meaning
of ``not necessarily consistently connected'').

\subsection{}

The following is the main definition in CbD. 
\begin{defn}
A system is \emph{noncontextual} if it has a \emph{multimaximally
connected coupling}. Otherwise it is \emph{contextual}. 
\end{defn}

The terminology in the definition is deciphered in the next three
sections (\ref{subsec:expl1}-\ref{subsec:expl3}) 

\subsection{}\label{subsec:expl1}

Let 
\begin{equation}
\mathcal{R}=\left\{ R_{q}^{c}:c\in C,q\in Q,q\prec c\right\} \label{eq:systemR}
\end{equation}
be a system (with $Q$ and $C$ being sets of contents and contexts,
respectively). A \emph{coupling} of this system is a correspondingly
labeled set of random variables
\begin{equation}
\mathcal{S}=\left\{ S_{q}^{c}:c\in C,q\in Q,q\prec c\right\} ,\label{eq:couplingS}
\end{equation}
such that 
\begin{lyxlist}{00.00.0000}
\item [{(a)}] its components are jointly distributed; 
\item [{(b)}] its context-wise marginals are distributed in the same way
as the corresponding subsets of the original system. 
\end{lyxlist}
That is, for any $c\in C$,
\begin{equation}
\left\{ R_{q}^{c}:q\in Q,q\prec c\right\} \sim\left\{ S_{q}^{c}:q\in Q,q\prec c\right\} .
\end{equation}

\subsection{}\label{subsec:expl2}

For instance, the jointly distributed random variables
\begin{equation}
\begin{array}{|c|c|c|c||c|}
\hline \begin{array}{c}
\\
\\
\end{array}S_{1}^{1} & S_{2}^{1} &  &  & c^{1}\\
\hline \begin{array}{c}
\\
\\
\end{array} & S_{2}^{2} & S_{3}^{2} & S_{4}^{2} & c^{2}\\
\hline \begin{array}{c}
\\
\\
\end{array}S_{1}^{3} &  & S_{3}^{3} &  & c^{3}\\
\hline \begin{array}{c}
\\
\\
\end{array}S_{1}^{4} &  &  & S_{4}^{4} & c^{4}\\
\hline \begin{array}{c}
\\
\\
\end{array}S_{1}^{5} & S_{2}^{5} & S_{3}^{5} &  & c^{5}\\
\hline\hline \begin{array}{c}
\\
\\
\end{array}q_{1} & q_{2} & q_{3} & q_{4} & \mathcal{S}
\\\hline \end{array}\label{eq:coupling example}
\end{equation}
form a coupling of system $\mathcal{R}$ in (\ref{eq:exmaple1}) if
the corresponding row-wise distributions in $\mathcal{S}$ and $\mathcal{R}$
coincide. 

\subsection{}\label{subsec:expl3}

Now, the coupling $\mathcal{S}$ is \emph{multimaximally connected}
if, for any two content-sharing random variables $S_{q}^{c},S_{q}^{c'}$
it contains, the probability of $S_{q}^{c}=S_{q}^{c'}$ is maximal
among all possible couplings of $\mathcal{R}$. Equivalently, the
probability of $S_{q}^{c}=S_{q}^{c'}$ in a multimaximal coupling
is maximal given the marginal distributions of $S_{q}^{c}\sim R_{q}^{c}$
and $S_{q}^{c'}\sim R_{q}^{c'}$. For system (\ref{eq:exmaple1})
with coupling (\ref{eq:coupling example}), this means the simultaneous
maximization of the probabilities of
\[
S_{1}^{1}=S_{1}^{3},S_{1}^{1}=S_{1}^{4},S_{1}^{3}=S_{1}^{4},S_{3}^{2}=S_{3}^{3},\ldots
\]

\subsection{}

Two couplings of the same system that have the same distribution are
considered equivalent and are not distinguished. In other words, the
domain probability space of a coupling with a given distribution can
be chosen arbitrarily. The most economic choice of the domain space
is the distribution itself, with all random variables being defined
as the componentwise projections of the identity function on this
space. That is, (\ref{eq:couplingS}) can be viewed as the identity
function on the probability space
\begin{equation}
\left(X_{S},\Sigma_{S},\mu_{S}\right)=\left\{ \left\{ -1,1\right\} ^{\left|\prec\right|},2^{\left\{ -1,1\right\} ^{\left|\prec\right|}},\mu_{S}:2^{\left\{ -1,1\right\} ^{\left|\prec\right|}}\rightarrow\left[0,1\right]\right\} ,
\end{equation}
where $\left|\prec\right|$ is the cardinality of the relation $\prec$,
and $\mu_{S}$ is the probability measure defined by the (joint) probability
mass function
\[
\Pr\left[S_{q}^{c}=s_{q}^{c}:c\in C,q\in Q,q\prec c\right],
\]
with $s_{q}^{c}=-1,1$. For uncountably infinite $Q$ and/or $C$
the definition should be modified in well-known ways, but we have
agreed to consider finite $Q,C$ only. 

\subsection{}

Note that any jointly distributed set of random variables is a random
variable. So the set of all variables within context $c$ can be written
as 
\[
\mathcal{R}^{c}=\left\{ R_{q}^{c}:q\in Q,q\prec c\right\} =R^{c},
\]
and a coupling $\mathcal{S}$ in (\ref{eq:couplingS}) can be written
as a random variable $S$.

\section{STC in the CbD Language}

\subsection{}

The language of STC is very different, and it manages to avoid even
mentioning random variables. Thus, our example system (\ref{eq:exmaple1})
could be represented in STC as{\small{}
\begin{equation}
\begin{array}{|c|c|c|c||c|}
\hline p_{\bar{1}\bar{1}}^{1}\begin{array}{c}
\\
\\
\end{array} & p_{\bar{1}1}^{1} & p_{1\bar{1}}^{1} & p_{11}^{1} & \left(q_{1},q_{2}\right)\\
\hline \begin{array}{c|c}
\begin{array}{c}
\\
\\
\end{array}p_{\bar{1}\bar{1}\bar{1}}^{2} & p_{\bar{1}\bar{1}1}^{2}\end{array} & \begin{array}{c|c}
\begin{array}{c}
\\
\\
\end{array}p_{\bar{1}1\bar{1}}^{2} & p_{\bar{1}11}^{2}\end{array} & \begin{array}{c|c}
\begin{array}{c}
\\
\\
\end{array}p_{1\bar{1}\bar{1}}^{2} & p_{1\bar{1}1}^{2}\end{array} & \begin{array}{c|c}
\begin{array}{c}
\\
\\
\end{array}p_{11\bar{1}}^{2} & p_{111}^{2}\end{array} & \left(q_{2},q_{3},q_{4}\right)\\
\hline \begin{array}{c}
\\
\\
\end{array}p_{\bar{1}\bar{1}}^{3} & p_{\bar{1}1}^{3} & p_{1\bar{1}}^{3} & p_{11}^{3} & \left(q_{1},q_{3}\right)\\
\hline \begin{array}{c}
\\
\\
\end{array}p_{\bar{1}\bar{1}}^{4} & p_{\bar{1}1}^{4} & p_{1\bar{1}}^{4} & p_{11}^{4} & \left(q_{1},q_{4}\right)\\
\hline \begin{array}{c|c}
\begin{array}{c}
\\
\\
\end{array}p_{\bar{1}\bar{1}\bar{1}}^{5} & p_{\bar{1}\bar{1}1}^{5}\end{array} & \begin{array}{c|c}
\begin{array}{c}
\\
\\
\end{array}p_{\bar{1}1\bar{1}}^{5} & p_{\bar{1}11}^{5}\end{array} & \begin{array}{c|c}
\begin{array}{c}
\\
\\
\end{array}p_{1\bar{1}\bar{1}}^{5} & p_{1\bar{1}1}^{5}\end{array} & \begin{array}{c|c}
\begin{array}{c}
\\
\\
\end{array}p_{11\bar{1}}^{5} & p_{111}^{5}\end{array} & \left(q_{1},q_{2},q_{3}\right)
\\\hline \end{array}\:,
\end{equation}
}where $p^{i}$ is the probability distribution in the $i$th context,
and the subscripts represent combinations of values $1$ and $-1\equiv\bar{1}$.
The values $\left(-1,-1\right)$ in $p_{\bar{1}\bar{1}}^{1}$ are
not values of the contents $\left(q_{1},q_{2}\right)$ (called in
STC ``measurements'', ``observables'', or simply ``variables'').
They are values of the random variables not being mentioned. 

\subsection{}

The reason this does not lead to complications is that STC is committed
to requiring that a system of random variable amenable to contextuality
analysis be \emph{strongly consistently connected}. This means that
for any pair of contexts $c,c'$, we have
\begin{equation}
\left\{ R_{q}^{c}:q\prec c,q\prec c'\right\} \sim\left\{ R_{q}^{c'}:q\prec c,q\prec c'\right\} ,
\end{equation}
i.e., the joint distributions for identically subscripted random variables
are identical. Thus, in our example (\ref{eq:exmaple1}), 
\[
\left\{ R_{1}^{5},R_{2}^{5}\right\} \sim\left\{ R_{1}^{1},R_{2}^{1}\right\} 
\]
and 
\[
\left\{ R_{2}^{5},R_{3}^{5}\right\} \sim\left\{ R_{2}^{2},R_{3}^{2}\right\} .
\]
 Abramsky and colleagues consider this property fundamental (see,
e.g., Ref. \cite{AbtramskyBrandSavochking2014}), and it is indeed
indispensable if one is to use the language of sheafs. 

\subsection{}

It should be noted, however, that in quantum-physical experiments
even simple consistent connectedness (which is obviously implied by
strong one) is routinely violated (see, e.g., Ref. \cite{KujDzhLar2015},
and for more references, Ref. \cite{Dzh2019}). In some non-physical
applications, notably in human behavior, inconsistent connectedness
is a universal rule \cite{CervDzhSQ,DzhZhaKuj2016}.

\subsection{}

Note that random variables may, in particular, be deterministic, i.e.
they may attain a single value with probability 1. If the requirement
of consistent connectedness is applied to such variables, then it
translates into any two content-sharing variables being equal to one
and the same value,
\begin{equation}
R_{q}^{c}\equiv r\Longleftrightarrow R_{q'}^{c}\equiv r.
\end{equation}
 We will see later, in Section \ref{sec:Deterministic-Systems}, that
this constraint creates a difficulty when STC deals with deterministic
systems.

\subsection{}

The property of strong consistent connectedness allows Abramsky and
colleagues to define a context simply by the set of contents measured
together. Thus, in the example (\ref{eq:exmaple1}), $c^{1}$ would
be defined as the context in which we measure $\left\{ q_{1},q_{2}\right\} $,
$c^{2}$ as the context in which we measure $\left\{ q_{2},q_{3},q_{4}\right\} $,
etc. In STC, there cannot be distinct contexts with the same set of
random variables in them, because their joint distributions would
have to be the same.

\subsection{}\label{subsec:exmapleC2}

To illustrate the effect of the restriction imposed by STC on the
CbD framework, consider the system
\begin{equation}
\begin{array}{|c|c||c|}
\hline \begin{array}{c}
\\
\\
\end{array}R_{1}^{1} & R_{2}^{1} & c^{1}\\
\hline \begin{array}{c}
\\
\\
\end{array}R_{1}^{2} & R_{2}^{2} & c^{2}\\
\hline\hline \begin{array}{c}
\\
\\
\end{array}q_{1} & q_{2} & \mathcal{C}_{2}
\\\hline \end{array}\:.\label{eq:cyclic2}
\end{equation}
In STC, it can only represent the same context repeated twice, which
makes the system trivially noncontextual. By contrast, in CbD, this
so-called cyclic system of rank 2 is the smallest nontrivial system,
one that can be contextual or noncontextual depending on the distributions
involved \cite{KujDzhProof2016}. In particular, the difference between
the two contexts may be the order in which the two contents are measured
($q_{1}\rightarrow q_{2}$ and $q_{2}\rightarrow q_{1}$) \cite{DzhZhaKuj2016}. 

\subsection{}

The assumption of consistent connectedness (not necessarily strong
one) simplifies the definition of (non)contextuality. 
\begin{defn}[Equivalent of STC definition]
A consistently connected system is noncontextual if it has an\emph{
identically connected coupling}. Otherwise it is \emph{contextual}. 
\end{defn}

A coupling $\mathcal{S}$ is identically connected if, for any $q,c,c'$
such that $q\prec c,c'$, 
\begin{equation}
\Pr\left[S_{q}^{c}=S_{q}^{c'}\right]=1.
\end{equation}

\subsection{}

This is, clearly, a special case of a multimaximal coupling: the maximal
probability of $S_{q}^{c}=S_{q}^{c'}$ in such a coupling is 1 if
and only if $R_{q}^{c}\sim R_{q}^{c'}$. For all practical purposes,
it allows one to think of the random variables $R_{q}^{c}$ as ``context-independent,''
and many authors would even denote them as $R_{q}$. This is, however,
this is a dubious practice that leads to a logical contradiction \cite{DzhKuj2017Fortsch,Dzh2019}.
STC avoids this difficulty by not mentioning random variables at all,
and systematically labeling probability distributions by their contexts
\cite{Dzh2019}. 

\subsection{}

Definition 2 does not explicitly require that consistent connectedness
of the system be strong. This, however, makes little difference if
one is only interested in determining whether a system is contextual.
It is easy to show the following.
\begin{thm}
A simply consistently connected system that is not strongly consistently
connected is contextual.
\end{thm}

\begin{proof}
In an identically connected coupling of the system, $\Pr\left[S_{q}^{c}=S_{q}^{c'}\right]=1$
and $\Pr\left[S_{q'}^{c}=S_{q'}^{c'}\right]=1$ imply 
\[
\Pr\left[\left(S_{q}^{c},S_{q'}^{c}\right)=\left(S_{q}^{c'},S_{q'}^{c'}\right)\right]=1,
\]
which is only possible if $\left(R_{q}^{c},R_{q'}^{c}\right)\sim\left(R_{q}^{c'},R_{q'}^{c'}\right)$.
\end{proof}
One can think of the strong consistent connectedness requirement in
STC as a provision excluding this ``guaranteed'' contextuality from
consideration. 

\section{Contextual Fraction}

\subsection{}

There are several possible ways of measuring the degree of contextuality
in CbD \cite{KUjDzh2019Measures}, but in STC the measure of choice
is contextual fraction. I present it, as everything else in this paper,
in the language of CbD. We need a few general probabilistic notions
first.

\subsection{}

An \emph{incomplete probability space}, or \emph{$\alpha$-probability
space} (where $0\leq\alpha\leq1$) is defined as a measure space $\left(X,\Sigma,\mu\right)$
with $\mu\left(X\right)=\alpha$. The meaning of the components of
the space (set $X$, sigma-algebra $\Sigma,$ and sigma-additive measure
$\mu$) is standard. Any measurable function $Z$ defined on this
space is called an \emph{incomplete (random) variable}, or a \emph{(random)
$\alpha$-variable}. It is essentially an ordinary random variable:
for any measurable set $D$ in the codomain of $Z$, $\Pr\left[Z\in D\right]$
is defined as $\mu\left(Z^{-1}\left(D\right)\right)$ and referred
to as the probability of $Z$ falling in $D$. The only difference
is that if $D=Z\left(X\right)$ then $\Pr\left[Z\in D\right]=\alpha$.
The rest of the concepts related to \emph{$\alpha$-}variables (e.g.,
their joint distribution) are the same as for true random variables
(with $\alpha=1$). In Feller's classical monograph \cite{Feller1968}
incomplete random variables are called ``defective''. Other terms,
such as ``improper'', are used too.

\subsection{}

Let $0\leq\alpha\leq\beta\leq1$, and let $Z_{\alpha}$ and $Z_{\beta}$
be an \emph{$\alpha$-}variable and a \emph{$\beta$-}variable, respectively,
with the same codomain. We say that $Z_{\alpha}$ is \emph{majorized}
by $Z_{\beta}$, if for every measurable set $D$ in their common
codomain, 
\begin{equation}
\Pr\left[Z_{\alpha}\in D\right]\leq\Pr\left[Z_{\beta}\in D\right].
\end{equation}
We write then 
\begin{equation}
Z_{\alpha}\apprle Z_{\beta}.
\end{equation}

\subsection{}

An \emph{incomplete (or $\alpha$-) coupling} of a system of random
variables $\mathcal{R}$ is a correspondingly indexed set 
\begin{equation}
^{\alpha}\mathcal{S}=\left\{ ^{\alpha}S_{q}^{c}:c\in C,q\in Q,q\prec c\right\} 
\end{equation}
 of jointly distributed $\alpha$-variables such that, for any context
$c\in C$, 
\begin{equation}
\left\{ ^{\alpha}S_{q}^{c}:q\in Q,q\prec c\right\} \apprle\left\{ R_{q}^{c}:q\in Q,q\prec c\right\} .
\end{equation}
An \emph{$\alpha$-}coupling is identically connected if 
\begin{equation}
\Pr\left[^{\alpha}S_{q}^{c}\not={}^{\alpha}S_{q}^{c'}\right]=0,
\end{equation}
for any $q\prec c,c'$. Clearly, an identically connected \emph{$\alpha$-}coupling
may exist only for a consistently connected system, and the latter
is noncontextual if and only if it has an identically connected \emph{$1$-}coupling. 

\subsection{}

The following theorem allows one to introduce a measure of contextuality.
\begin{thm}
\label{thm:pip}Any consistently connected system has an identically
connected \emph{$\alpha_{\max}$-}coupling ($0\leq\alpha_{\max}\leq1$),
such that the system has no identically connected \emph{$\alpha$-}couplings
with \textup{$\alpha>\alpha_{\max}$.}
\end{thm}

This property can be proved by employing the linear programming representation
of the relation between an $\alpha$-coupling and the context-wise
distributions in the system. This representation is essentially the
same as one routinely used in both STC \cite{AbramBarbMans2017,AbramskyBrand2011}
and CbD \cite{KUjDzh2019Measures,DzhCerKuj2017}. 
\begin{proof}
We represent the system $\mathcal{R}$ and an \emph{$\alpha$-}coupling
by probability vectors $\mathbf{r}$ and $\mathbf{s}$, respectively,
such that
\begin{equation}
\mathbf{B}\mathbf{s}\leq\mathbf{r}\textnormal{ (componentwise)},
\end{equation}
with the following meaning of the terms. The entries of $\mathbf{r}$
are context-wise joint probabilities 
\begin{equation}
\Pr\left[R_{q}^{c}=r_{q}:q\in Q,q\prec c\right],\label{eq:bunches}
\end{equation}
across all $c\in C$ and all combinations of $r_{q}=+1/-1$. The entries
of $\mathbf{z}$ are joint probabilities
\begin{equation}
\Pr\left[^{\alpha}S_{q}^{c}=s_{q}:c\in C,q\in Q,q\prec c\right],\label{eq:hidden}
\end{equation}
across all combinations of $s_{q}=+1/-1$. $\mathbf{B}$ is a Boolean
matrix (``incidence matrix'' in Ref. \cite{AbramskyBrand2011})
with rows indexed by the values of $c$, and, for each $c$, by the
combinations of $r_{q}$-values in (\ref{eq:bunches}). Its columns
are indexed by the combinations of $s_{q}$-values in (\ref{eq:hidden}).
A cell of $\mathbf{B}$ is filled with 1 if its $s_{q}$-combination
contains its $r_{q}$-combination for the corresponding $q$-values,
otherwise the cell is filled with 0. The set of all \emph{$\alpha$-}couplings
of $\mathcal{R}$ is represented by the (obviously nonempty) polytope

\begin{equation}
\mathbb{Z}=\{\mathbf{s}:\mathbf{B}\mathbf{s}\leq\mathbf{\mathbf{r}},\mathbf{s}\geq0,\mathbf{1}\cdot\mathbf{s}\leq1\}.\label{eq: CNTF_Z}
\end{equation}
Every linear functional, including $\mathbf{1}\cdot\mathbf{s}$, attains
its extrema within this polytope, and the maximum value of $\mathbf{1}\cdot\mathbf{s}$
is taken as $\alpha_{\max}$. 
\end{proof}
\subsection{}

The quantity $\alpha_{\max}$ is called \emph{noncontextual fraction},
and $1-\alpha_{\max}$ is called \emph{contextual fraction}. It is
easy to see that if $\alpha_{\max}=1$, then $\mathbf{B}\mathbf{s}=\mathbf{\mathbf{r}}$,
and the system is noncontextual. Otherwise it is contextual, and the
contextual fraction is a natural measure of the degree of contextuality. 

\subsection{}

Abramsky and colleagues single out the case $\alpha_{\max}=0$, calling
such systems \emph{strongly contextual}. In strongly contextual systems,
every possible combination of $s_{q}$-values has the probability
of zero. This is the situation one encounters with the Kochen-Specker
systems \cite{KS1967} and with the Popescu-Rohrlich boxes \cite{PR1994}.

\section{Consistified Systems}

\subsection{}\label{subsec:first}

Can STC, and contextual fraction in particular, be generalized to
arbitrary, generally inconsistently connected, systems? It turns out
this can be done by a simple procedure that converts an inconsistently
connected system $\mathcal{R}$ into a \emph{contextually equivalent}
consistently connected one, $\mathcal{R}^{\ddagger}$. Contextual
equivalence means that $\mathcal{R}$ is contextual if and only if
so is $\mathcal{R}^{\ddagger}$, and that if $\mathcal{R}$ is consistently
connected, then $\mathcal{R}^{\ddagger}$ has the same value of contextual
fraction.

\subsection{}

The discussion of the procedure of \emph{consistification} is helped
by two additional CbD terms. Let us call all random variables sharing
a context a \emph{bunch}, 
\begin{equation}
\mathcal{R}^{c}=\left\{ R_{q}^{c}:q\in Q,q\prec c\right\} =R^{c},\label{eq:bunch}
\end{equation}
and all random variables sharing a content a \emph{connection},
\begin{equation}
\mathcal{R}_{q}=\left\{ R_{q}^{c}:c\in C,q\prec c\right\} .\label{eq:connection}
\end{equation}
 The terminology is intuitive: a bunch is jointly distributed, and
different bunches are disjoint, but the fact that some random variables
in different contexts have the same content creates connections between
the bunches. 

\subsection{}

Within a connection (\ref{eq:connection}) the random variables are
stochastically unrelated, but they can be coupled by
\begin{equation}
\mathcal{T}_{q}=\left\{ T_{q}^{c}:c\in C,q\prec c\right\} =T_{q},\label{eq:couplingT}
\end{equation}
and among all such couplings one can seek a \emph{multimaximal coupling},
one that maximizes the probabilities for all equalities 
\begin{equation}
T_{q}^{c}=T_{q}^{c'},q\prec c,c'.
\end{equation}

\subsection{}

If, as we have agreed, all random variables in the system are dichotomous,
then we have the following result, proved in \cite{DzhKuj2017_2.0}. 
\begin{thm}
\label{thm:pop}Any connection has one and only one multimaximal coupling.
$\mathcal{T}_{q}$ is a multimaximal coupling of the connection $\mathcal{R}_{q}$
if and only if any subset of $\mathcal{T}_{q}$ is a maximal coupling
of the corresponding subset of $\mathcal{R}_{q}$.
\end{thm}

The second statement means that, for any part 
\[
\mathcal{R}'_{q}=\left\{ R_{q}^{c}:c\in\left\{ c_{1},\ldots,c_{k}\right\} \subseteq C,q\prec c\right\} ,
\]
of the connection (\ref{eq:connection}), the event
\[
T_{q}^{c_{1}}=\ldots=T_{q}^{c_{k}}
\]
has the maximal possible probability among all possible couplings
of $\mathcal{R}'_{q}$ (or, equivalently, given the marginal distributions
of $T_{q}^{c_{i}}\sim R_{q}^{c_{i}}$, $i=1,\ldots,k$).

\subsection{}\label{subsec:firstdreason}

The theorem above is one of the two reasons why CbD subjects any set
of measurements to dichotomization before making it a system of random
variables amenable to contextuality analysis. The consistification
procedure to be described is based on the possibility to find unique
multimaximal couplings for all connections. 

\subsection{}\label{subsec:secondreason}

For completeness, I should mention the second, and main, reason for
the dichotomization of all random variables: it prevents the otherwise
possible situation when coarse-graining of the random variables in
a noncontextual system makes it contextual \cite{DzhCerKuj2017}.
Clearly, a good theory of contextuality should not have this property.
Note that dichotomization does not lose any information extractable
from random variables before they are dichotomized. It does not even
increase the size of a system, if size is measured by the cardinality
of the supports of the system's bunches.

\subsection{}

The idea of consistification is to treat the multimaximal couplings
of connections as if they were additional bunches. This is implicit
in any CbD-based algorithm for establishing or measuring contextuality
\cite{KujDzhProof2016,KujDzhLar2015,KUjDzh2019Measures}. Explicitly,
however, it was first described by Amaral and coauthors \cite{AmaralDuarteOliveira2018}.
However, Amaral and coauthors use maximal couplings instead of the
multimaximal ones (as we did in the older version of CbD, e.g., in
\cite{DzhKuj2016}), and they allow for multivalued variables. The
difference between the two types of couplings of a set $\left\{ X_{1},\ldots,X_{n}\right\} $
is that in the multimaximal coupling $\left\{ Y_{1},\ldots,Y_{n}\right\} $
we maximize probabilities of all equalities $Y_{i}=Y_{j}$ (whence
it follows, by Theorem \ref{thm:pop}, that we also maximize the probability
of $Y_{i_{1}}=Y_{i_{2}}=\ldots=Y_{i_{k}}$ for any subset of $\left\{ Y_{1},\ldots,Y_{n}\right\} $),
whereas a maximal coupling $\left\{ Z_{1},\ldots,Z_{n}\right\} $
only maximizes the probability of the single chain equality $Z_{1}=Z_{2}=\ldots=Z_{n}$.
Maximal couplings generally are not unique, even for dichotomous variables
(if there are more than two of them). Since measures of (non)contextuality
generally depend on what couplings are being used, the approach advocated
in Ref. \cite{AmaralDuarteOliveira2018} faces the problem of choice.
In addition, a system declared noncontextual using maximal rather
than multimaximal couplings may have contextual subsystems, obtained
by dropping some of the variables. I consider this possibility highly
undesirable for a theory of contextuality.

\subsection{}

The particular consistification scheme presented below is an elaboration
of one described to me by Janne Kujala (personal communication, November
2018).

\subsection{}

Given an arbitrary system $\mathcal{R}$, the new system $\mathcal{R}^{\ddagger}$
has a set $Q^{\ddagger}$ of ``new'' contents, a set $C^{\ddagger}$
of ``new'' contexts, and a ``new'' is-measured-in relation $\prec^{\ddagger}$.
The corresponding constructs in the original system, $Q$, $C$, and
$\prec$, will be called ``old''.

\subsection{}

For each random variable $R_{j}^{i}$ in $\mathcal{R}$ we form a
new content, denoted $q_{j}^{i}$. The set of all new contents is
\begin{equation}
Q^{\ddagger}=\left\{ q_{j}^{i}:c^{i}\in C,q_{j}\in Q,q_{j}\prec c^{i}\right\} .
\end{equation}
The number of the new contents is the cardinality of $\prec$, which
cannot exceed $\left|C\times Q\right|$. 

\subsection{}

New contexts are formed as the set 
\[
C^{\ddagger}=C\sqcup Q,
\]
and their number is $\left|C\right|+\left|Q\right|$.

\subsection{}

The new is-measured-in relation is
\begin{equation}
\prec^{\ddagger}=\left\{ \left(q_{j}^{i},c^{i}\right):c^{i}\in C,q_{j}\in Q,q_{j}\prec c^{i}\right\} \sqcup\left\{ \left(q_{j}^{i},q_{j}\right):c^{i}\in C,q_{j}\in Q,q_{j}\prec c^{i}\right\} .
\end{equation}
That is, a new content $q_{j}^{i}$ is measured in the new contexts
$c^{i}$ and $q_{j}$ only.

\subsection{}

Each $\left(q_{j}^{i},c^{i}\right)$-cell contains the old random
variables $R_{j}^{i}$. The new bunch
\begin{equation}
R^{i}=\left\{ R_{j}^{i}:q_{j}^{i}\in Q^{\ddagger},q_{j}^{i}\prec^{\ddagger}c^{i}\right\} 
\end{equation}
coincides with the old bunch
\begin{equation}
R^{i}=\left\{ R_{j}^{i}:q_{j}\in Q,q_{j}\prec c^{i}\right\} .
\end{equation}

\subsection{}

Each $\left(q_{j}^{i},q_{j}\right)$-cell contains a new random variable
$V_{j}^{i}$ whose distribution is the same as that of $R_{j}^{i}$.
The bunch
\begin{equation}
V^{j}=\left\{ V_{j}^{i}:q_{j}^{i}\in Q^{\ddagger},q_{j}^{i}\prec^{\ddagger}q_{j}\right\} 
\end{equation}
is the multimaximal coupling of the old connection
\begin{equation}
\mathcal{R}_{j}=\left\{ R_{j}^{i}:c^{i}\in C,q_{j}\prec c^{i}\right\} .
\end{equation}

\subsection{}

Using our examples (\ref{eq:exmaple1}) and (\ref{eq:cyclic2}), the
corresponding consistified systems are
\begin{equation}
{\normalcolor \begin{array}{|c|c|c|c|c|c|c|c|c|c|c|c||c|}
\hline \begin{array}{c}
\\
\\
\end{array}R_{1}^{1} & R_{2}^{1} &  &  &  &  &  &  &  &  &  &  & c^{1}\\
\hline \begin{array}{c}
\\
\\
\end{array} &  & R_{2}^{2} & R_{3}^{2} & R_{4}^{2} &  &  &  &  &  &  &  & c^{2}\\
\hline \begin{array}{c}
\\
\\
\end{array} &  &  &  &  & R_{1}^{3} & R_{3}^{3} &  &  &  &  &  & c^{3}\\
\hline \begin{array}{c}
\\
\\
\end{array} &  &  &  &  &  &  & R_{1}^{4} & R_{4}^{4} &  &  &  & c^{4}\\
\hline \begin{array}{c}
\\
\\
\end{array} &  &  &  &  &  &  &  &  & R_{1}^{5} & R_{2}^{5} & R_{3}^{5} & c^{5}\\
\hline \begin{array}{c}
\\
\\
\end{array}V_{1}^{1} &  &  &  &  & V_{1}^{3} &  & V_{1}^{4} &  & V_{4}^{4} &  &  & q_{1}\\
\hline \begin{array}{c}
\\
\\
\end{array} & V_{2}^{1} & V_{2}^{2} &  &  &  &  &  &  &  & V_{2}^{5} &  & q_{2}\\
\hline \begin{array}{c}
\\
\\
\end{array} &  &  & V_{3}^{2} &  &  & V_{3}^{3} &  &  &  &  & V_{3}^{5} & q_{3}\\
\hline \begin{array}{c}
\\
\\
\end{array} &  &  &  & V_{4}^{2} &  &  &  & V_{4}^{4} &  &  &  & q_{4}\\
\hline\hline \begin{array}{c}
\\
\\
\end{array}q_{1}^{1} & q_{2}^{1} & q_{2}^{2} & q_{3}^{2} & q_{4}^{2} & q_{1}^{3} & q_{3}^{3} & q_{1}^{4} & q_{4}^{4} & q_{1}^{5} & q_{2}^{5} & q_{3}^{5} & \mathcal{A}^{\ddagger}
\\\hline \end{array}}\label{eq: consistified}
\end{equation}
and
\begin{equation}
\begin{array}{|c|c|c|c||c|}
\hline \begin{array}{c}
\\
\\
\end{array}R_{1}^{1} & R_{2}^{1} &  &  & c^{1}\\
\hline \begin{array}{c}
\\
\\
\end{array} &  & R_{1}^{2} & R_{2}^{2} & c^{2}\\
\hline \begin{array}{c}
\\
\\
\end{array}V_{1}^{1} &  & V_{1}^{2} &  & q_{1}\\
\hline \begin{array}{c}
\\
\\
\end{array} & V_{2}^{1} &  & V_{2}^{2} & q_{2}\\
\hline\hline \begin{array}{c}
\\
\\
\end{array}q_{1}^{1} & q_{2}^{1} & q_{1}^{2} & q_{2}^{2} & \mathcal{C}_{2}^{\ddagger}
\\\hline \end{array}\:.\label{eq:consistifiedC2}
\end{equation}

\subsection{}

Note the following properties of all consistified systems.
\begin{lyxlist}{00.00.0000}
\item [{1.}] Bunches corresponding to different old contexts, $c^{i},c^{i'},$
are disjoint.
\item [{2.}] Bunches corresponding to different old contents, $q_{j},q_{j'}$,
are disjoint.
\item [{3.}] A bunch corresponding to an old content, $q_{j}$, and a bunch
corresponding to an old context, $c^{i}$, have at most one connection
between them, $\left\{ R_{j}^{i},V_{j}^{i}\right\} $.
\item [{4.}] The connection corresponding to any new content $q_{j}^{i}$
contains precisely two random variables, $R_{j}^{i}$ and $V_{j}^{i}$,
with the same distribution.
\end{lyxlist}
\subsection{}

Property 3 above means that in a consistified system the notions of
simple consistent connectedness and strong consistent connectedness
coincide. As mentioned earlier, in Section \ref{subsec:exmapleC2},
system $\mathcal{C}_{2}$ in (\ref{eq:cyclic2}) will only be considered
in STC if the two bunches $\left\{ R_{1}^{1},R_{2}^{1}\right\} $
and $\left\{ R_{1}^{2},R_{2}^{2}\right\} $ are identically distributed,
which would make this system trivial. By contrast, if one adopts the
CbD-based consistification of $\mathcal{C}_{2}$, in (\ref{eq:consistifiedC2}),
its STC analysis will be the same as CbD's.

\subsection{}

The following fact ensures that the generalization of the contextual
fraction coincides with the original one in the case of consistently
connected systems.
\begin{thm}
If a system $\mathcal{R}$ is consistently connected, then its contextual
fraction is the same as that of $\mathcal{R}^{\ddagger}$. (In particular,
$\mathcal{R}$ is contextual if and only if so is $\mathcal{R}^{\ddagger}$.)
\end{thm}

\begin{proof}
Immediately follows from the observation that any state of an $\alpha$-coupling
in which the values corresponding to $R_{q}^{c}$ and $V_{q}^{c}$
are different has the probability zero. 
\end{proof}
\subsection{}

For completeness, I formulate the following as a formal statement.
Recall the definition of contextual equivalence in Section \ref{subsec:first}.
\begin{thm}
Any system $\mathcal{R}$ is contextually equivalent to its consistification
$\mathcal{R}^{\ddagger}$.
\end{thm}

\begin{proof}
Follows from the previous theorem, and the obvious fact that $\mathcal{R}$
and $\mathcal{R}^{\ddagger}$ have the same linear programming representation
(see Ref. \cite{KUjDzh2019Measures} for a detailed description of
the latter). 
\end{proof}

\section{Deterministic Systems\label{sec:Deterministic-Systems}}

\subsection{}

Deterministic systems can be viewed as systems of random variables
whose distributions attain specific values with probability 1. In
CbD, therefore, they are treated as a special case of systems of random
variables, with the following general result.
\begin{thm}
Any deterministic system is noncontextual.
\end{thm}

\begin{proof}
A deterministic system 
\[
\mathcal{R}=\left\{ R_{q}^{c}\equiv r_{q}^{c}:c\in C,q\in Q,q\prec c\right\} ,
\]
where $\equiv$ means equality with probability 1, has a single overall
coupling,
\[
\mathcal{S}=\left\{ S_{q}^{c}\equiv r_{q}^{c}:c\in C,q\in Q,q\prec c\right\} ,
\]
with all $S_{q}^{c}$ defined on an arbitrary probability space. Since
$\left\{ S_{q}^{c}\equiv r_{q}^{c},S_{q}^{c'}\equiv r_{q}^{c'}\right\} $
is the only coupling of $\left\{ R_{q}^{c}\equiv r_{q}^{c},R_{q}^{c'}\equiv r_{q}^{c'}\right\} $,
the probability of $S_{q}^{c}=S_{q}^{c'}$ (0 or 1) is maximal possible,
whence $\mathcal{S}$ is multimaximally connected.
\end{proof}
\subsection{}

This simple observation seems to put CbD at odds with STC, where the
theoretical ideas formulated in algebraic and topological terms are
not restricted to random variables. Consider, e.g., two deterministic
systems that have the same $\prec$-format as the system $\mathcal{C}_{2}$
in example (\ref{eq:cyclic2}): 
\begin{equation}
\begin{array}{|c|c||c|}
\hline \begin{array}{c}
\\
\\
\end{array}R_{1}^{1}\equiv1 & R_{2}^{1}\equiv-1 & c^{1}\\
\hline \begin{array}{c}
\\
\\
\end{array}R_{1}^{2}\equiv1 & R_{2}^{2}\equiv-1 & c^{2}\\
\hline\hline \begin{array}{c}
\\
\\
\end{array}q_{1} & q_{2} & \mathcal{C}_{2.1}
\\\hline \end{array}\textnormal{ and }\begin{array}{|c|c||c|}
\hline \begin{array}{c}
\\
\\
\end{array}R_{1}^{1}\equiv1 & R_{2}^{1}\equiv-1 & c^{1}\\
\hline \begin{array}{c}
\\
\\
\end{array}R_{1}^{2}\equiv1 & R_{2}^{2}\equiv1 & c^{2}\\
\hline\hline \begin{array}{c}
\\
\\
\end{array}q_{1} & q_{2} & \mathcal{C}_{2.2}
\\\hline \end{array}\:.\label{eq:2determ}
\end{equation}
System $\mathcal{C}_{2.1}$ is consistently connected, which in a
deterministic system means it is strongly consistently connected.
It is therefore trivially noncontextual. System $\mathcal{C}_{2.2}$
is inconsistently connected. Strictly speaking, therefore, the original
STC analysis should not be applicable to this system, as it violates
the fundamental assumption underlying STC. 

\subsection{}

If we use the extended version of STC, with the help of consistification,
we get
\begin{equation}
\begin{array}{|c|c|c|c||c|}
\hline \begin{array}{c}
\\
\\
\end{array}R_{1}^{1}\equiv1 & R_{2}^{1}\equiv-1 &  &  & c^{1}\\
\hline \begin{array}{c}
\\
\\
\end{array} &  & R_{1}^{2}\equiv1 & R_{2}^{2}\equiv1 & c^{2}\\
\hline \begin{array}{c}
\\
\\
\end{array}V_{1}^{1}\equiv1 &  & V_{1}^{2}\equiv1 &  & q_{1}\\
\hline \begin{array}{c}
\\
\\
\end{array} & V_{2}^{1}\equiv-1 &  & V_{2}^{2}\equiv1 & q_{2}\\
\hline\hline \begin{array}{c}
\\
\\
\end{array}q_{1}^{1} & q_{2}^{1} & q_{1}^{2} & q_{2}^{2} & \mathcal{C}_{2.2}^{\ddagger}
\\\hline \end{array}\:.
\end{equation}
This system is trivially noncontextual by the STC/CbD definition. 

\subsection{}

This reasoning would apply to any deterministic system: if it is consistently
connected (or consistified), it is trivially noncontextual, and if
it is inconsistently connected, STC should place it outside its sphere
of applicability (or consistify it). In other words, STC with consistification
and CbD treat deterministic systems identically (finding them noncontextual).
There is, of course, a simple way out: to complement STC with the
additional stipulation that all inconsistently connected systems are
contextual. STC would then have to allow for contextual systems whose
degree of contextuality cannot be measured by contextual fraction.
I do not think this simple way out is intellectually satisfactory. 

\subsection{}

Here is a good place to mention that CbD treats inconsistent connectedness
and contextuality as fundamentally different concepts. Inconsistent
connectedness, i.e. the difference in the distributions of $R_{q}^{c}$
and $R_{q}^{c'}$, is interpreted as the result of direct influences
of the contexts upon the measurements. In the case of physical systems,
one can say that some elements of the contexts $c$ and $c'$ differently
affect (in the causal sense) the measurement of the content $q$.
In quantum physics this is reflected by such notions as ``signaling''
or (a better term) ``disturbance''. Contextuality, by contrast,
is non-causal, and reflects the differences between random variables
$R_{q}^{c}$ and $R_{q}^{c'}$ that are above and beyond the differences
in their distributions. 

\subsection{}

This interpretation is philosophically based on the \emph{no-conspiracy
principle} \cite{CervDzhSQ}, according to which in ``not-precariously-unstable''
and ``not-deliberately-contrived'' systems, no differences in the
direct influences exerted by the elements of context are hidden. Being
hidden means that these differences are present but are not reflected
in the differences of the distributions. For instance, if $R_{q}^{c}$
and $R_{q}^{c'}$ attain values $1$ and $-1$ with probability $\frac{1}{2}$
each, and if $c'$ by some causal mechanism reverses (multiplies by
$-1$) each value of $R_{q}^{c'}$, then this influence will be hidden,
as it will not affect the distribution of $R_{q}^{c'}$. The no-conspiracy
principle says this should not be expected to happen, and if it does,
should be expected to disappear by slight modifications of the experimental
set-up. The principle is closely related to the ``no-fine-tuning''
principle advocated by Cavalcanti \cite{Cavalcanti2018} (see a detailed
analysis of these principles by Jones in Ref. \cite{MattJones2019}).

\subsection{}

With this in mind, let us consider an especially elegant application
of STC to an inherently deterministic system, described in Ref. \cite{Abramskyetal2017}.
This is a system whose contents are statements referencing each other's
truth value and forming a version of the Liar antinomy. I will consider
the version with three statements, although any larger number will
be analyzed similarly:

\begin{equation}
\begin{array}{|c|c|c||c|}
\hline \begin{array}{c}
\\
\\
\end{array}R_{1}^{1} & R_{2}^{1} &  & c^{1}\\
\hline \begin{array}{c}
\\
\\
\end{array} & R_{2}^{2} & R_{3}^{2} & c^{2}\\
\hline \begin{array}{c}
\\
\\
\end{array}R_{1}^{3} &  & R_{3}^{3} & c^{3}\\
\hline\hline \begin{array}{c}
\\
\\
\end{array}q_{1}=\textnormal{"}q_{2}\textnormal{ is true"} & q_{2}=\textnormal{"}q_{3}\textnormal{ is true"} & q_{3}=\textnormal{"}q_{1}\textnormal{ is false"} & \mathcal{L}_{3}
\\\hline \end{array}\:.
\end{equation}
The contexts combine the statements one of which references the other,
and the $R_{q}^{c}$ is the truth value (1 or $-1$) of statement
$q$ in context $c$. 

\subsection{}

We could have considered the smaller system
\begin{equation}
\begin{array}{|c|c||c|}
\hline \begin{array}{c}
\\
\\
\end{array}R_{1}^{1} & R_{2}^{1} & c^{1}\\
\hline \begin{array}{c}
\\
\\
\end{array}R_{1}^{2} & R_{2}^{2} & c^{2}\\
\hline\hline \begin{array}{c}
\\
\\
\end{array}q_{1}=\textnormal{"}q_{2}\textnormal{ is true"} & q_{2}=\textnormal{"}q_{1}\textnormal{ is false"} & \mathcal{L}_{2}
\\\hline \end{array}\:,
\end{equation}
representing a more familiar classical form of the antinomy, but the
contexts in $\mathcal{L}_{3}$ are easier to interpret, as the direction
of inference there need not be specified. The interpretation is even
more complicated with the classical form $q=\lyxmathsym{\textquotedblleft}q\textnormal{ is false}$'',
although the reasoning below is still applicable. 

\subsection{}

We can posit that each statement in a given context should have one
definitive truth value, and this makes $\mathcal{L}_{3}$ a deterministic
system. In Ref. \cite{Abramskyetal2017} this system is characterized
as strongly contextual, based on the impossibility to assign the truth
values in a context-independent way. However, we know that the original
version of STC is predicated on the assumption of strong consistent
connectedness, whereas any deterministic realization of $\mathcal{L}_{3}$
(precisely because no context-independent assignment of truth values
exists) is inconsistently connected. Consider one of the eight such
deterministic versions, corresponding to the usual conceptualization
of the Liar Antinomy: 

\begin{equation}
\begin{array}{|c|c|c||c|}
\hline \begin{array}{c}
\\
\\
\end{array}R_{1}^{1}\equiv1 & R_{2}^{1}\equiv1 &  & c^{1}\\
\hline \begin{array}{c}
\\
\\
\end{array} & R_{2}^{2}\equiv1 & R_{3}^{2}\equiv1 & c^{2}\\
\hline \begin{array}{c}
\\
\\
\end{array}R_{1}^{3}\equiv-1 &  & R_{3}^{3}\equiv1 & c^{3}\\
\hline\hline \begin{array}{c}
\\
\\
\end{array}q_{1}=\textnormal{"}q_{2}\textnormal{ is true"} & q_{2}=\textnormal{"}q_{3}\textnormal{ is true"} & q_{3}=\textnormal{"}q_{1}\textnormal{ is false"} & \mathcal{L}_{3.1}
\\\hline \end{array}\:.
\end{equation}
The arguments related to system $\mathcal{C}_{2.2}$ in (\ref{eq:2determ})
apply here fully. I see no reasonable way a system like $\mathcal{L}_{3.1}$
can be treated as contextual, either in CbD or in STC. 

\subsection{}

There is, however, another way of looking at system $\mathcal{L}_{3}$.
It seems to be very much in the spirit of how it is treated in Ref.
\cite{Abramskyetal2017}. Moreover, it corresponds to the traditional
presentation of the Liar antinomy: suppose $R_{1}^{1}\equiv1$, then
it follows that $R_{1}^{1}\equiv-1$; now suppose $R_{1}^{1}\equiv-1$,
then it follows that $R_{1}^{1}\equiv1$. With the stipulation that
each statement in a given context should have one definitive truth
value, $\mathcal{L}_{3}$ can indeed be just one of the eight deterministic
(and inconsistently connected) systems of which $\mathcal{L}_{3.1}$
is one. However, we do not know which of these eight systems to choose,
and we can consider all eight of them as variants of $\mathcal{L}_{3}$:
\[
\begin{array}{c}
\begin{array}{|c|c|c||c|}
\hline 1 & 1 &  & c^{1}\\
\hline  & 1 & 1 & c^{2}\\
\hline -1 &  & 1 & c^{3}\\
\hline\hline q_{1} & q_{2} & q_{3} & \mathcal{L}_{3.1}
\\\hline \end{array}\\
\\
\begin{array}{|c|c|c||c|}
\hline -1 & -1 &  & c^{1}\\
\hline  & -1 & -1 & c^{2}\\
\hline 1 &  & -1 & c^{3}\\
\hline\hline q_{1} & q_{2} & q_{3} & \mathcal{L}_{3.2}
\\\hline \end{array}\\
\\
\begin{array}{|c|c|c||c|}
\hline -1 & -1 &  & c^{1}\\
\hline  & 1 & 1 & c^{2}\\
\hline -1 &  & 1 & c^{3}\\
\hline\hline q_{1} & q_{2} & q_{3} & \mathcal{L}_{3.3}
\\\hline \end{array}\\
\\
etc.
\end{array}
\]

\subsection{}

One can assign Bayesian (or epistemic) probabilities to these possibilities,
a natural choice here being to assign them uniformly. This renders
the system probabilistic in the epistemic sense, with
\begin{equation}
\left\langle R_{q}^{c}\right\rangle _{B}=0,
\end{equation}
for all the ``epistemically-random'' variables (indicated by the
subscript $B$), and
\begin{equation}
\left\langle R_{1}^{1}R_{2}^{1}\right\rangle _{B}=\left\langle R_{2}^{2}R_{3}^{2}\right\rangle _{B}=-\left\langle R_{3}^{3}R_{1}^{3}\right\rangle _{B}=1.
\end{equation}
This is a Bayesian analogue of a rank 3 cyclic system that is consistently
connected and forms a Popescu-Rohrlich box. Its contextuality, both
in CbD and STC, is maximal. In particular, when measured by contextual
fraction, it is strong ($\alpha_{\max}=0$), in accordance with how
Abramsky and colleagues view it.

\subsection{}

This Bayesian procedure can be applied to any deterministic system
with more than one possible deterministic realization. The procedure
will render the system quasi-probabilistic and, at least in all the
simple cases I can think of, consistently connected and contextual.
More work is needed to elaborate this approach. 

\section{Conclusion}

\subsection{}

We have seen that STC can be extended to apply to inconsistently connected
systems, using CbD-based multimaximal couplings to consistify these
systems. We have also seen that the Bayesian rendering of the deterministic
systems with multiple possible realizations allows STC to circumvent
the difficulty associated with inconsistent connectedness of each
of these realizations. It simultaneously extends CbD to such systems
and allows CbD to treat them in the spirit of STC, forming thereby
another bridge between the two theories. 

\subsection{}

Together, the consistification and the Bayesian treatment make STC
and CbD essentially coextensive, with a major proviso: one has to
agree to represent all measurement outcomes in a system as sets of
jointly distributed dichotomous random variables. Dichotomization
of a system is always possible, so it is more of a language choice
than a restriction of applicability. Dealing only with dichotomous
variables allows one to avoid a variety of difficulties \cite{DzhCerKuj2017},
but no proof exists that they could not be avoided by other means. 

\subsection{}

Finally, nothing in this paper implies that CbD can be replaced with
STC, or vice versa. Each of the two theories has its own aims and
means. Thus, logical aspects of contextuality, especially in the possibilistic
proofs of contextuality, are significantly more salient in STC than
CbD, adding to the former's aesthetic elegance. Perhaps the use of
the Bayesian/epistemic random variables, as discussed above, might
offer CbD a way to ``catch up'' in this respect. STC in turn might
benefit from using the language of random variables for proof purposes.
For instance, the fact that the existence of a hidden variable model
for a system of random variables is equivalent to the existence of
their joint distribution (from which it follows, in particular, that
nonlocality is a special case of contextuality) is true almost by
definition if the language of random variables is used explicitly.
It is a non-trivial, perhaps even surprising fact, however, if one
considers the systems of random variables in terms of their distributions
only \cite{Fine1982,AbramskyBrand2011}. It would be good if the equivalences
established in this paper helped the two theories to more freely borrow
from each other's native languages, follow each other's directions
of research, and use each other's proof techniques.

\end{document}